\documentclass[a4paper,onecolumn,11pt]{quantumarticle}
\pdfoutput=1
\usepackage[T1]{fontenc}
\usepackage[utf8]{inputenc}
\usepackage{amsmath,amssymb,amsthm,mathtools}
\usepackage[numbers,sort&compress]{natbib}
\usepackage{booktabs,array}
\usepackage{needspace}
\usepackage{graphicx,xcolor}
\usepackage{hyperref}
\definecolor{refcolor}{rgb}{0.067,0.42,0.42}
\hypersetup{colorlinks=true,allcolors=refcolor,
	pdftitle={Mid-circuit measurements and feedforward implement every unitary exactly in constant depth},
	pdfauthor={Chenfeng Cao}}
\newcommand{\C}{\mathbb C}
\newcommand{\R}{\mathbb R}
\newcommand{\Z}{\mathbb Z}
\newcommand{\ket}[1]{\lvert #1\rangle}
\newcommand{\bra}[1]{\langle #1\rvert}
\newcommand{\reg}[1]{\mathsf{#1}}
\newcommand{\Us}{\widetilde U}
\newcommand{\opnorm}[1]{\lVert #1\rVert_{\mathrm{op}}}
\newcommand{\fnorm}[1]{\lVert #1\rVert_{\mathrm F}}
\newcommand{\T}{\mathsf T}
\newcommand{\Fou}{\mathcal F}
\newtheorem{theorem}{Theorem}
\newtheorem{lemma}[theorem]{Lemma}
\newtheorem{corollary}[theorem]{Corollary}
\newtheorem{proposition}[theorem]{Proposition}
\allowdisplaybreaks[1]

\makeatletter
\patchcmd{\@maketitle}{\hyphenpenalty=50000}%
{\raggedright\hyphenpenalty=50000}{}%
{\PackageWarning{main}{title patch failed}}
\makeatother

\begin{document}
	\title{Mid-circuit measurements and feedforward implement every unitary exactly in constant depth}
	\author{Chenfeng Cao}
	\email{caocf@hku.hk}
	\affiliation{HK Institute of Quantum Science \& Technology, The University of Hong Kong, Hong Kong, China}

	\date{October 2026}
	
	\begin{abstract}
		Mid-circuit measurements with feedforward reduce the optimal worst-case quantum depth for exact $n$-qubit unitary synthesis from $\Theta(n)$ to $O(1)$, using arbitrary one- and two-qubit gates and unrestricted ancillary space. Measurements and resets without feedforward remain subject to the circuit causal-cone bound. Our deterministic construction uses single-qubit measurements, $O(1)$ rounds of parity feedforward, and $O(n4^n\log(n+2))$ total qubits, with all ancillas starting and ending in $\ket0$. Classical processing is not counted in quantum depth. We teleport only the matrix error of an approximation by Nehoran and Yuen. An exactly prepared ancillary state encodes this error, whose action is added back by coherent interference. At operator-norm error $2^{-n}$, one step of amplitude amplification suffices. Coherent fanout circuits achieve the same constant-depth bound, while circuits using only one- and two-qubit gates attain depth $\Theta(n)$ with $O(4^n\log^2(n+2))$ gates. Generic targets require $\Omega(4^n/(h+1))$ qubits at depth $h$, even under unrestricted classical control. Our depth--width trade-offs meet this bound within a factor $O(n^2)$ for $1\le h\le2^n$.
	\end{abstract}
	\maketitle
	
	\section{Introduction}
	
	Mid-circuit measurements with feedforward allow entangled ancillas
	and classical corrections to replace sequential quantum operations.
	Teleportation and measurement-based computation provide familiar examples
	\cite{bennett1993,gottesman1999,raussendorf2001,raussendorf2003}.
	Measurements with feedforward also prepare long-range entangled states in shallow
	circuits~\cite{piroli2021,lu2022,tantivasadakarn2023} and underlie
	dynamic-circuit experiments~\cite{baumer2024}. Measurements also reduce
	the depth of commuting circuits~\cite{cao2026} and provide non-Gaussian
	resources for fermionic sampling~\cite{caotang2026}.
	A generic $n$-qubit unitary requires exponentially many
	gates~\cite{shende2006}, but ancillary qubits allow many of them
	to run in parallel, as quantified by depth--ancilla
	trade-offs~\cite{sun2023,zhangli2022,yuanzhang2023}. We ask whether
	feedforward of measurement outcomes can reduce exact
	synthesis of an arbitrary unitary on an unknown input to constant
	quantum depth, given unrestricted ancillary space.
	
	We work first with coherent circuits that allow fanout (a CNOT with
	many targets), then replace
	each fanout gate by a constant-depth gadget using measurements and
	parity feedforward~\cite{browne2011,buhrman2024,zi2025}.
	Exact constant-depth state preparation is known for circuits with
	generalized Toffoli gates, with or without fanout
	\cite{rosenthal2026,ggj2026,ssw2026}, and for circuits with measurements and
	feedforward~\cite{zi2025}.
	Unitary synthesis must transform every unknown input superposition
	coherently.
	For arbitrary unitaries,
	Rosenthal obtained exact depth $\widetilde O(2^{n/2})$ with one- and
	two-qubit gates and $O(2^{n/2})$ with fanout and generalized Toffoli
	gates~\cite{rosenthal2026}, where $\widetilde O$ suppresses logarithmic
	factors. Strelchuk, Subramanian, and Weisz gave
	exact constant-depth constructions for permutations and diagonal
	unitaries, leaving the general case open~\cite{ssw2026}. Nehoran and
	Yuen showed that every unitary can be approximated at constant
	depth with fanout~\cite{nehoran2026}. Replacing fanout by measurements
	and feedforward gives the corresponding approximate adaptive
	implementation. For $d=2^n$, their width is
	$\widetilde O(d^2\log^4(d/\varepsilon))$, so accuracy
	$\varepsilon=1/d$ still needs only $\widetilde O(d^2)$ qubits
	\cite[Theorem~5.1]{nehoran2026}. Here we give exact implementations at
	finite width, allowing arbitrary continuous gate parameters as in
	unitary synthesis~\cite{barenco1995}.
	
	The \emph{block} of a circuit is the operator obtained by fixing all
	ancilla inputs and outputs to $\ket0$.
	Gate teleportation through the Choi state of $U$ requires a correction
	$UP^\dagger U^\dagger$ for a Bell outcome labeled by the Pauli operator
	$P$~\cite{gottesman1999}. For a general $U$, this correction need not
	be local. Coherently retaining the identity-outcome component gives
	the block $U/d$~\cite[Sec.~7.2]{irani2022}, so standard amplitude
	amplification takes $\Theta(d)$ rounds~\cite{brassard2002,berry2014}.
	We teleport only the residual
	of an approximation. If its block $\Us$ satisfies
	$\opnorm{U-\Us}\le1/d$, the residual $E=U-\Us$
	satisfies
	\begin{equation}
		\nu=\sqrt d\,\fnorm E\le d\,\opnorm E\le1.
		\label{eq:key-intro}
	\end{equation}
	For $E\ne0$, teleportation yields the block $E/\nu$. Combining
	$\Us$ and $E/\nu$ coherently with relative coefficients $1$ and
	$\nu$ produces $U/(1+\nu)$.
	The repaired block is exactly proportional to $U$, with normalization
	$1+\nu\le2$ and the same success amplitude for every input.
	A flag sets this amplitude to $1/2$, so one amplification step
	succeeds on every input superposition (Fig.~\ref{fig:repair}).
	The residual-state preparation is independent of the input. Its
	cost accounts for the leading width in our construction. We call this procedure
	\emph{residual repair}.
	
	Without feedforward, an output depends only on operations in its
	backward circuit causal cone. The cone grows at most exponentially
	with depth, limiting how many parameters can influence the output.
	Feedforward supplies additional classical control paths between
	quantum layers, while parameter counting still limits the total
	circuit size.
	With unrestricted ancillary space, the optimal worst-case exact
	depth is $\Theta(n)$ without feedforward, even when measurements and
	resets are allowed. Parity feedforward reduces it to $O(1)$.
	The ordinary-circuit upper bound answers Rosenthal's Question~1.2
	and meets the joint size and depth goals of Question~1.8
	\cite{rosenthal2026}.
	Together, these results settle the
	exact versions of the ordinary and fanout questions in Open
	Problem~17~\cite{rosenthalthesis}. Combined with the fanout elimination
	of Ref.~\cite{ssw2026}, they also resolve the unrestricted-size
	question of Strelchuk, Subramanian, and Weisz.
	Applying the construction to smaller unitaries in a cosine--sine
	decomposition gives depth--width trade-offs.
	
	\begin{figure*}[t]
		\centering
		\includegraphics[width=\textwidth]{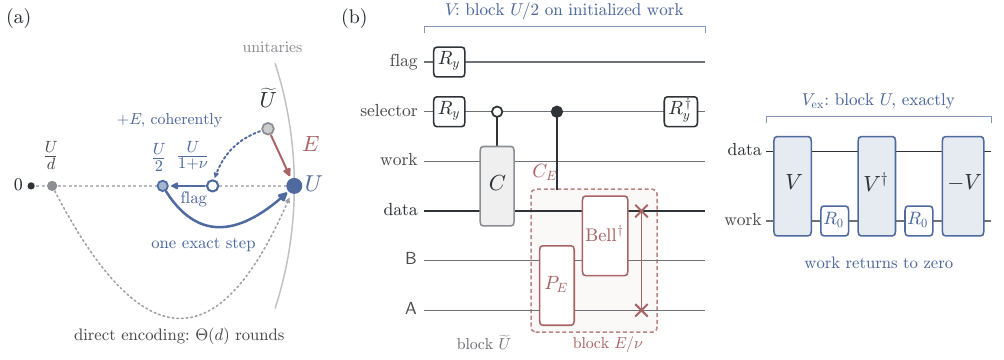}
		\caption{Residual repair. (a)~Schematic picture of the blocks. Direct Choi encoding
			(dotted route) requires $\Theta(d)$ amplification rounds. Adding the
			residual instead produces $U/(1+\nu)$, which a flag rescales to
			$U/2$ before one exact amplification step. (b)~The circuit $V$ and
			the amplifier. The selector chooses
			between the approximate circuit $C$ (open control) and the residual
			circuit $C_E$ (closed control). Inside $C_E$, the preparation circuit satisfies
			$P_E\ket0=\ket e_{\reg{AB}}\ket0_{\rm prep}$; preparation work is
			omitted from the drawing. The circuit $C_E$ then applies the
			inverse Bell transform to $\reg B$ and the data, and swaps $\reg A$
			into the data position. All registers except the data start in
			$\ket0$. The operator $R_0=I-2\Pi$ in
			Eq.~\eqref{eq:amplification} flips the phase of the joint $\ket0$
			state of all non-data qubits, including preparation work.}
		\label{fig:repair}
	\end{figure*}
	
	\section{Model and main result}
	\label{sec:model}
	
	We consider $n$ input qubits and ancillas (work qubits) initialized in
	$\ket0$. We allow arbitrary one- and two-qubit gates with all-to-all
	connectivity and continuous parameters specified exactly.
	A layer consists of operations on disjoint
	supports. The \emph{quantum depth} $h$ counts gate and single-qubit
	measurement layers, including resource preparation, corrections, and
	resetting; the \emph{width} $w$ counts all qubits, including the input.
	Classical processing of measurement outcomes is not charged to quantum
	depth. Our adaptive constructions use a fixed
	gate and measurement schedule; outcomes determine only Pauli
	corrections and resets through affine parities, called \emph{parity
		feedforward} below. Our measurements are in the computational basis;
	the lower bounds allow arbitrary single-qubit measurement bases.
	Without feedforward, all operations and the choice of output qubits
	are independent of the measurement record; unconditional resets are
	still allowed. The general width lower bound also allows outcome-dependent
	schedules and arbitrary classical processing. Synthesis is nonuniform:
	the target and gate parameters are specified classically, with no
	bound on classical compilation time.
	
	An adaptive implementation of a unitary is exact and deterministic
	if, conditioned on any measurement record of nonzero probability,
	it implements the target unitary channel and resets all work
	qubits to $\ket0$. A coherent circuit is \emph{clean} if it maps
	$\ket\psi\ket0$ to $U\ket\psi\ket0$ for every $\ket\psi$. By linearity,
	both notions cover inputs entangled with a reference. Besides one- and
	two-qubit gates, coherent circuits may use the fanout gate
	\begin{equation}
		F_q:\ket{x,y_1,\ldots,y_q}
		\longmapsto\ket{x,y_1\oplus x,\ldots,y_q\oplus x}
		\label{eq:fanout}
	\end{equation}
	and generalized Toffoli gates with arbitrarily many controls and one
	target. Coherent circuits using only one- and two-qubit gates are called
	\emph{ordinary}; allowing fanout gives \emph{fanout circuits}
	\cite{green2002,hoyer2005}.
	We build coherent circuits with generalized Toffoli gates and then
	compile them into fanout or adaptive circuits. We write $d=2^n$ and use
	$\opnorm{\cdot}$ and $\fnorm{\cdot}$ for the operator and Frobenius
	norms.
	
	\Needspace{14\baselineskip}
	\begin{theorem}[Exact constant-depth synthesis]
		\label{thm:main}
		Every $n$-qubit unitary has an exact deterministic implementation of
		quantum depth $O(1)$ using one- and two-qubit gates,
		single-qubit measurements, $O(1)$ rounds of parity feedforward, and $O(nd^2\log(n+2))$
		qubits; all ancillas start and end in $\ket0$. The same width suffices
		for an exact clean constant-depth fanout circuit. Conversely, for Haar-almost every unitary,
		an exact constant-depth implementation in either model requires
		$\Omega(d^2)$ qubits, even if the classical controller is unrestricted.
	\end{theorem}
	Sections~\ref{sec:repair} and~\ref{sec:synthesis} give the construction;
	Proposition~\ref{prop:width} gives the lower bound.
	With the same local gates and product-state ancillas, measurements
	without feedforward still require depth $\Omega(n)$ for generic
	targets, as shown in Theorem~\ref{thm:ordinary}.
	
	For a chain of $N$ qubits, every nearest-neighbor circuit of depth
	$O(\log N)$ has an exact constant-depth implementation in the all-to-all
	model with parity
	feedforward and $\operatorname{poly}(N)$ qubits
	(Corollary~\ref{cor:log}). Its gates can be grouped into two layers
	of disjoint $O(\log N)$-qubit blocks, to which Theorem~\ref{thm:main}
	applies. This covers $O(\log N)$ periods of a one-dimensional Floquet
	circuit whose period has a constant number of nearest-neighbor
	gate layers.
	
	\section{Residual repair}
	\label{sec:repair}
	
	Throughout this section, circuits use one- and two-qubit gates,
	fanout, and generalized Toffoli gates. For a circuit $C$ on $n$ data
	qubits and work registers initialized in $\ket0$, its \emph{block} is
	\begin{equation}
		\Us=(I\otimes\bra0)\,C\,(I\otimes\ket0),
		\label{eq:block}
	\end{equation}
	a $d\times d$ matrix, not necessarily unitary. The circuit $C$ is
	a block encoding of $\Us$: measuring its work qubits and obtaining
	all zeros would apply $\Us$ to the data before normalization.
	The repair keeps this component coherent.
	Circuits are padded to common work registers, and
	$\Pi=I\otimes\ket0\!\bra0$ denotes the projector onto initialized
	work qubits. If an approximation is specified up to global phase,
	we choose the representative of $U$ closest to $\Us$. A known
	one-qubit phase restores any prescribed matrix representative
	before a controlled version is used.
	
	A circuit of depth $h$ and width $w$ in this gate set can be inverted,
	or controlled by one qubit, in depth $O(h+1)$ and width $O(w)$ with
	clean additional work qubits. One fans the control out to a separate copy for each gate
	of a layer; a controlled generalized Toffoli gate gains one control,
	and a controlled fanout computes, fans out, and uncomputes the
	conjunction of its source and the control~\cite{barenco1995,
		rosenthal2026}. The reflection $I-2\Pi$ is a generalized Toffoli gate
	conjugated by single-qubit gates, so it too has constant depth.
	
	\Needspace{15\baselineskip}
	\begin{theorem}[Residual repair]
		\label{thm:repair}
		Let $U\in\mathrm U(d)$, and let $C$ be a coherent circuit of one-
		and two-qubit gates, fanout, and generalized Toffoli gates, of depth
		$h$ and width $w$, whose block $\Us$ satisfies
		\begin{equation}
			\nu=\sqrt d\,\fnorm{U-\Us}\le1,
			\label{eq:tolerance}
		\end{equation}
		which holds in particular when $\opnorm{U-\Us}\le1/d$. When
		$E=U-\Us\neq0$, suppose the $2n$-qubit state proportional to
		$\sum_{i,j}E_{ij}\ket i\ket j$ has an exact clean
		preparation in the same gate set of depth $h_E$ and width $w_E$. Then
		$U$ has an exact clean circuit of depth $O(h+h_E+1)$ and width
		$O(w+w_E)$ in the same gate set.
	\end{theorem}
	
	\begin{proof}
		Put $E=U-\Us$. If $E=0$, the projected output $\Pi C\ket\psi\ket0
		=U\ket\psi\ket0$ already has norm $\|\psi\|$, so $C$ itself is exact
		and clean. Assume $E\neq0$.
		
		\emph{Teleporting the residual.} On two $n$-qubit registers $\reg A$
		and $\reg B$, prepare
		\begin{equation}
			\ket e_{\reg{AB}}=\frac1{\fnorm E}\sum_{i,j=0}^{d-1}
			E_{ij}\ket i_{\reg A}\ket j_{\reg B}.
			\label{eq:residual-state}
		\end{equation}
		Apply to the data and $\reg B$ the inverse Bell transform, a layer of
		CNOTs followed by Hadamards that maps $d^{-1/2}\sum_j\ket j\ket j$ to
		$\ket0\ket0$, and swap the data with $\reg A$. Call the whole circuit,
		including the state preparation, $C_E$. Its block acts as
		\begin{equation}
			\begin{aligned}
				(I\otimes\bra0)\,C_E\,(I\otimes\ket0)\ket\psi
				&=\frac{1}{\sqrt d\,\fnorm E}\sum_{i,j}E_{ij}\,\psi_j\ket i\\
				&=\frac{E\ket\psi}{\nu}.
			\end{aligned}
			\label{eq:residual-block}
		\end{equation}
		
		\emph{Adding the residual.} Prepare a selector qubit in
		$(\ket0+\sqrt\nu\ket1)/\sqrt{1+\nu}$, apply $C$ if it is $\ket0$ and
		$C_E$ if it is $\ket1$, and unprepare it. With the selector counted as
		a work qubit, the block becomes
		\begin{equation}
			\begin{aligned}
				\frac{1}{1+\nu}\,\Us+\frac{\nu}{1+\nu}\frac{E}{\nu}
				&=\frac{\Us+E}{1+\nu}\\
				&=\frac{U}{1+\nu}.
			\end{aligned}
			\label{eq:addition}
		\end{equation}
		Unpreparing the selector makes the two contributions interfere in
		the component with all work qubits zero.
		Apply $R_y(\theta)$ to a flag initialized in $\ket0$, with
		$\cos(\theta/2)=(1+\nu)/2$, to rescale the block to $U/2$.
		Count the flag as a work qubit and call the
		resulting circuit $V$.
		
		\emph{One exact amplification step.} The block $U/2$ has success
		amplitude $1/2=\sin(\pi/6)$ for every input. One amplification step
		triples this angle to $\pi/2$, making success certain. Define
		\begin{equation}
			V_{\mathrm{ex}}=-V(I-2\Pi)V^\dagger(I-2\Pi)V.
			\label{eq:amplification}
		\end{equation}
		\Needspace{9\baselineskip}
		Expanding the reflections and using
		$\Pi V\Pi=\tfrac12U\otimes\ket0\!\bra0$ and
		$\Pi V^\dagger\Pi V\Pi=\Pi/4$, we obtain
		\begin{equation}
			\begin{aligned}
				V_{\mathrm{ex}}\Pi
				&=V\Pi+2\Pi V\Pi-4V\Pi V^\dagger\Pi V\Pi\\
				&=V\Pi+U\otimes\ket0\!\bra0-V\Pi\\
				&=U\otimes\ket0\!\bra0.
			\end{aligned}
			\label{eq:cubic}
		\end{equation}
		Thus the circuit implements $U$ and returns every work qubit to
		$\ket0$. Two controlled calls to $C$ and one to $C^\dagger$, and
		likewise for $C_E$, account for the stated resources.
	\end{proof}
	
	Linear combinations of unitaries~\cite{childswiebe2012} and oblivious
	amplitude amplification~\cite{berry2014,berry2015} are standard tools.
	The latter includes the one-step exact case, also used in
	Refs.~\cite[Lemma~6.11]{ssw2026}
	and~\cite[Proposition~2]{vasconcelos2026}.
	Novo and Berry improved a constant-accuracy Hamiltonian simulation
	with a truncated correction followed by amplitude
	amplification~\cite{novoberry2017}.
	Here the entire residual $E$ is encoded exactly as a state, and
	$\sqrt d\,\fnorm E\le1$ makes one amplification step sufficient.
	
	\section{Constant-depth ingredients}
	\label{sec:synthesis}
	
	\subsection{Approximation and state preparation}
	\label{sec:seed}
	
	To apply Theorem~\ref{thm:repair}, we use the approximation of
	Ref.~\cite{nehoran2026} at accuracy $1/d$. The implementation below
	gives the stated width. The input is encoded as a single excitation
	shared among auxiliary oscillator modes. Quadratic phase gates
	shear their phase space, while Fourier transforms exchange position
	and momentum. Three phase layers and two Fourier layers transform
	the single-excitation amplitudes before decoding
	\cite{nehoran2026,fiurasek2003}. Representing each mode on a finite grid gives
	an approximate qubit circuit with the following resource bound.
	
	\begin{lemma}[Approximate synthesis]
		\label{lem:seed}
		For every $U\in\mathrm U(d)$ and $0<\varepsilon<1$, there is a
		constant-depth fanout circuit whose block $\Us$ satisfies
		$\opnorm{U-\Us}\le\varepsilon$. It uses
		$O\bigl(d^2[\log\log(d/\varepsilon)]^2+d\log^4(d/\varepsilon)\bigr)$
		qubits.
	\end{lemma}
	
	Writing the quadratic phases in register bits produces commuting
	one- and two-qubit phase gates that fanout allows to run in parallel
	\cite[Theorem~3.2]{hoyer2005}.
	This reduces the phase-layer width of
	Ref.~\cite[Algorithm~4]{nehoran2026}; Fourier transforms and encoding
	contribute the second term. Supplementary Material~\ref{app:seed}
	proves the finite-grid error bound.
	
	\begin{lemma}[Exact state preparation]
		\label{lem:stateprep}
		Every $q$-qubit pure state has an exact clean preparation of constant depth
		with one- and two-qubit gates, fanout, and generalized Toffoli gates on
		$O(q2^q)$ qubits, and with a fanout circuit on
		$O(q^22^q)$ qubits.
	\end{lemma}
	
	Rosenthal established exact clean constant-depth state preparation
	with $\widetilde O(2^q)$ qubits~\cite[Theorem~4.1]{rosenthal2026}.
	Sharing prefix computations reduces the $O(4^q)$ workspace of
	Ref.~\cite[Theorem~4.2]{ssw2026} to $O(q2^q)$ with generalized
	Toffoli gates. Supplementary Material~\ref{app:stateprep} gives the
	construction and compilation.
	
	\begin{proof}[Coherent part of Theorem~\ref{thm:main}]
		
		Take the approximate circuit of Lemma~\ref{lem:seed} with $\varepsilon=1/d$; its width
		is $O(d^2\log^2(n+2))$. The residual state of
		Theorem~\ref{thm:repair} has $2n$ qubits, so Lemma~\ref{lem:stateprep}
		prepares it exactly in constant depth on $O(nd^2)$ qubits with
		generalized Toffoli gates, and Theorem~\ref{thm:repair} yields an exact
		clean constant-depth circuit of width $W=O(nd^2)$. The
		Toffoli compilation of Lemma~\ref{lem:toffoli}
		expands each width-$W$ layer to
		$O(W\log\log(W+3))=O(nd^2\log(n+2))$ qubits, since gates within a
		layer have disjoint supports.
	\end{proof}
	
	\subsection{From coherent synthesis to dynamic circuits}
	\label{sec:adaptive}
	
	Constant-depth fanout with parity feedforward is
	standard~\cite{broadbent2009,browne2011,buhrman2024,zi2025,foxman2026}.
	Related measurement-based fanout protocols have been demonstrated on quantum
	hardware~\cite{baumer2025}.
	Each target couples to its own ancilla in parallel, and the measurement
	outcomes determine the final local Pauli corrections.
	The seven-layer construction in Supplementary Material~\ref{app:fanout}
	uses $2q+1$ ancillas for $q$ targets, including preparation and reset.
	Since fanouts in a layer have disjoint supports, replacing them in a
	coherent circuit of depth $h$ and width $w$ gives
	\begin{equation}
		h_{\rm meas}\le7h,\qquad w_{\rm meas}\le3w,
		\label{eq:measurement-compiler}
	\end{equation}
	with at most $h$ feedforward rounds, proving the adaptive part of
	Theorem~\ref{thm:main}. Controls, adjoints, and amplification are
	completed before substitution, so each gadget replaces an
	uncontrolled fanout and its outcome-dependent phase remains global.
	
	\section{Depth--width trade-offs}
	\label{sec:bounds}
	
	\subsection{Parameter counting}
	
	For adaptive unitary synthesis, parameter counting
	\cite{shende2006,sun2023,yuanzhang2023} can be applied to a single
	nonzero measurement branch.
	
	\begin{proposition}[Generic depth--width lower bound]
		\label{prop:width}
		For Haar-almost every $n$-qubit unitary, every exact implementation of
		depth $h$, coherent with fanout and generalized Toffoli gates or
		adaptive with single-qubit measurements and arbitrary classical
		feedforward, uses $\Omega(d^2/(h+1))$ qubits.
	\end{proposition}
	
	\begin{proof}
		Refine every measurement and reset into individual Kraus outcomes, and
		contract the final ancillary outputs with computational-basis bras.
		This gives Kraus operators $K_\omega:\C^d\to\C^d$ of the implemented
		channel. A unitary channel has Choi rank one~\cite{choi1975}, so each
		nonzero $K_\omega$ equals $\lambda_\omega U$ for a scalar
		$\lambda_\omega\neq0$. Fix such a branch. Its gates, measurement bases,
		and wiring, including adaptively chosen parameters, are fixed, and at
		width $w$ and depth $h$ they contain $O((h+1)w)$ real parameters: each
		local operation contributes $O(1)$, while fanout and generalized
		Toffoli gates contribute none.
		
		For each wiring, $K_\omega$ is a smooth function of these parameters.
		On the open set $K_\omega\ne0$, the map to its projective class
		$[K_\omega]$ is smooth and locally Lipschitz, hence its image has
		Hausdorff dimension $O((h+1)w)$. The exactly realizable unitary
		classes lie in this image. There are finitely many wirings
		at fixed $h,w$. Haar measure on the projective unitary group is its
		smooth $(d^2-1)$-dimensional volume, so their union is Haar-null
		unless $(h+1)w=\Omega(d^2)$. A countable union over $h,w$
		preserves this conclusion and proves the converse in
		Theorem~\ref{thm:main}. Applying the same argument to a nonzero
		successful branch also covers heralded exact implementations.
	\end{proof}
	
	In an exact deterministic implementation of a unitary, the branch
	probability $|\lambda_\omega|^2$ is independent of the input.
	The measurement record therefore carries no information about the
	unknown state.
	
	\begin{figure}[tb]
		\centering
		\includegraphics[width=0.6\linewidth]{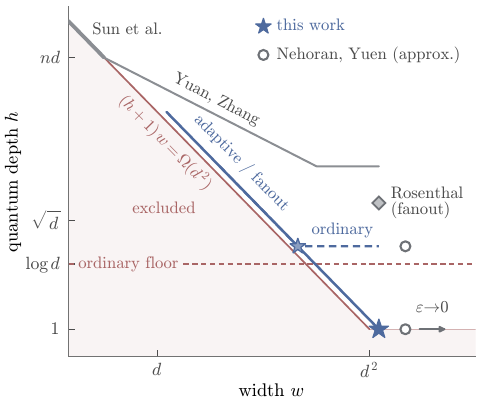}
		\caption{Schematic depth--width landscape for generic unitaries,
			with constants and logarithmic factors suppressed. Shading marks the
			parameter-counting bound on circuit spacetime volume. The linear-depth
			floor without feedforward follows from backward circuit causal cones.
			Solid and dashed blue curves show
			Theorems~\ref{thm:tradeoff} and~\ref{thm:ordinary-tradeoff}; dark and
			light stars mark the constant-depth and $O(n)$-depth constructions.
			Gray curves and the diamond are earlier exact constructions
			\cite{sun2023,yuanzhang2023,rosenthal2026}. Open circles show
			Ref.~\cite{nehoran2026} at $\varepsilon=1/d$, offset to the right for
			visibility; the arrow shows increasing accuracy at constant fanout
			depth. For smaller widths, Sun et al.'s bound
			$O(nd+d^2/w)$~\cite{sun2023} completes the frontier up to factors
			polynomial in $n$.}
		\label{fig:landscape}
	\end{figure}
	
	\subsection{Trading width for depth}
	
	A multiplexer applies a unitary selected by the computational-basis
	value of a control register. A cosine--sine decomposition through
	$r$ qubits expresses $U$
	as $O(2^r)$ multiplexers whose component unitaries act on at most $n-r+1$ qubits
	\cite{paigewei1994,mottonen2004,shende2006}. For each multiplexer, a
	one-hot control routes the input to the selected component. The
	controlled component circuits run in parallel; reversing the routing returns the work
	qubits to $\ket0$ for reuse by the next factor. The resulting bound is proved
	in Supplementary Material~\ref{app:tradeoff}.
	
	\Needspace{9\baselineskip}
	\begin{theorem}[Depth--width trade-off]
		\label{thm:tradeoff}
		For every $1\le h\le d$, every $n$-qubit unitary has an exact clean
		fanout circuit of depth $O(h)$ and width
		$O\bigl((d^2/h)\log^2(4d/h)\bigr)$. The same bounds hold for exact
		deterministic adaptive circuits with parity feedforward. The width is
		within a factor $O(\log^2(4d/h))$ of Proposition~\ref{prop:width}.
	\end{theorem}
	
	\subsection{Depth without feedforward}
	
	Replacing fanout in Theorem~\ref{thm:main} by CNOT trees gives
	ordinary depth $O(n)$. The direct construction in Supplementary
	Material~\ref{app:ordinary-tradeoff} reaches the same depth with fewer
	qubits and gates.
	
	\Needspace{10\baselineskip}
	\begin{theorem}[Optimal depth without feedforward]
		\label{thm:ordinary}
		Every $n$-qubit unitary has an exact clean circuit of one- and
		two-qubit gates with depth $O(n)$, width
		$O(d^2\log^2(n+2)/n)$, and $O(d^2\log^2(n+2))$ gates. For Haar-almost every unitary, every such circuit with any
		number of ancillas initialized in a product state has depth
		\begin{equation}
			h\ge2n-\log_2n-O(1).
			\label{eq:ordinary-lower}
		\end{equation}
		This lower bound also holds with single-qubit measurements and
		unconditional resets, provided no measurement outcome controls
		subsequent operations or the choice of output qubits.
	\end{theorem}
	
	\begin{proof}
		Take $w=\Theta(d^2\log^2(n+2)/n)$ in the direct construction of
		Supplementary Material~\ref{app:ordinary-tradeoff}; this lies in its
		range for sufficiently large $n$, and the finitely many remaining
		cases are absorbed into the constants. Substitution gives depth
		$O(n)$, and the gate count is at most depth times width.
		
		The causal-cone argument of Refs.~\cite{sun2023,yuanzhang2023} also
		applies to measurements and resets without feedforward. Averaging over outcomes
		turns these operations into single-qubit trace-preserving channels.
		Their adjoints preserve the identity, so operations outside the
		backward causal cones of the fixed output qubits cannot affect the
		output channel. The cones contain only $O(n2^h)$ local operations
		and product-state parameters.
		Up to relabeling ancillas there are
		finitely many cone topologies. Each describes a channel family of
		dimension $O(n2^h)$, whereas unitary channels have dimension $d^2-1$.
		Thus a generic target requires $n2^h=\Omega(d^2)$, giving
		Eq.~\eqref{eq:ordinary-lower}.
	\end{proof}
	
	This exponentially improves the previous exact depth
	$\widetilde O(\sqrt d)$ of Refs.~\cite[Corollary~4.6]{rosenthal2026}
	and~\cite[Theorem~3]{yuanzhang2023}.
	Rosenthal's accompanying optimality result concerns queries to
	oracles that prepare the columns of $U$~\cite[Theorem~2.3]{rosenthal2026}.
	Nehoran and Yuen achieve approximate ordinary depth
	$O(\log d+\log\log(d/\varepsilon))$~\cite[Theorem~5.1(1)]{nehoran2026};
	Theorem~\ref{thm:ordinary} removes the approximation error at depth $O(n)$.
	More generally, at width $O(w)$ for $w\ge Cdn$, the depth is
	$O(n+d^2\log^2(n+2)/w)$, within a factor $O(\log^2(n+2))$
	of the generic lower bound (Supplementary
	Material~\ref{app:ordinary-tradeoff} and Fig.~\ref{fig:landscape}).
	
	\section{Discussion}
	\label{sec:outlook}
	
	The construction interleaves a constant number of quantum stages
	with classical parity computations. Using two-input XOR gates, each
	feedforward round can be evaluated in depth $O(\log w)=O(n)$.
	Whether a single feedforward round can achieve width
	$\widetilde O(d^2)$ at constant quantum depth remains open.
	
	Coherent error-$1/3$ approximations suffice if available for every
	unitary: Supplementary Material~\ref{app:constant-accuracy} combines
	$O(n)$ such circuits in parallel and corrects the final residual
	exactly. Supplementary Material~\ref{sm:independent} constructs the
	block $U/2$ from convex combinations of permutations, giving a proof
	independent of Nehoran--Yuen at width $2^{O(2^n)}$.
	
	Decomposing the two-qubit gates of Theorem~\ref{thm:main} into single-qubit gates and CNOTs
	\cite{barenco1995}, and replacing fanout by the clean construction of
	Ref.~\cite[Lemma~6.12]{ssw2026}, leaves only
	single-qubit and generalized Toffoli gates, at a possibly doubly
	exponential width. That lemma uses Rosenthal's approximate
	construction~\cite{rosenthal2021} and one exact amplification step;
	Section~4 of Ref.~\cite{ssw2026} gives another route using Grier,
	Morris, and Wu~\cite{grier2026} together with Rosenthal.
	Polynomial-size fanout elimination remains open
	\cite{fang2006,nadimpalli2024,anshu2024}.
	
	Our constructions use continuous gate parameters. With Clifford+$T$
	gates and computational-basis measurements, generic exact synthesis
	is impossible~\cite{kliuchnikov2013,giles2013}.
	
	With fanout or parity feedforward, can the width be reduced to
	$O(d^2)$ at constant depth? Our construction has overheads from
	residual-state preparation, Toffoli compilation, and the quadratic
	phase layers. It also remains open whether an approximation-free
	construction can reach width $\widetilde O(d^2)$.
	
	\section*{AI disclosure}
	The author posed the problem and proposed the residual-repair strategy. The results and derivations were developed through iterative exchanges with GPT-6 Astra, which contributed substantially to the constructions and proofs and drafted much of the manuscript. The author selected the directions pursued and independently verified every proof. Numerical checks on small systems covered Eqs.~\eqref{eq:residual-block}, \eqref{eq:addition}, \eqref{eq:cubic}, and~\eqref{sm:eq:permutation-dilation}. The author takes full responsibility for the article.

	\bibliographystyle{quantum}
	\bibliography{refs}
	
	\clearpage
	\onecolumn
	\appendix
	\setcounter{theorem}{0}
	\renewcommand{\thetheorem}{S\arabic{theorem}}
	\renewcommand{\theHtheorem}{S\arabic{theorem}}
	\setcounter{equation}{0}
	\renewcommand{\theequation}{S\arabic{equation}}
	\renewcommand{\theHequation}{S\arabic{equation}}
	\setcounter{figure}{0}
	\renewcommand{\thefigure}{S\arabic{figure}}
	\renewcommand{\theHfigure}{S\arabic{figure}}
	\section*{Supplementary Material}
	
	\section{Gaussian approximate synthesis}
	\label{app:seed}
	
	We use Nehoran and Yuen's continuous identity
	\cite[Theorem~3.1]{nehoran2026}, which they relate to Fiur\'a\v sek's
	continuous-variable gate synthesis~\cite{fiurasek2003}.
	Following the wrapped-sampling approach discussed in
	Ref.~\cite[Remark~4.4]{nehoran2026}, we make the Fourier steps exact
	and bound the errors in the three phase steps. A bitwise expansion
	of the phases then gives the width in Lemma~\ref{lem:seed}.
	
	\subsection{Continuous identity}
	For $x\in\R^m$, let $Q_Af(x)=e^{ix^\T Ax/2}f(x)$ for real symmetric
	$A$, and let $\Fou$ be the unitary Fourier transform with kernel
	$(2\pi)^{-m/2}e^{-ip^\T x}$. The map
	\begin{equation}
		\alpha\longmapsto\phi_\alpha(x)
		=\sqrt2\,\pi^{-m/4}(\alpha^\T x)e^{-|x|^2/2}
		\label{eq:single-excitation}
	\end{equation}
	is an isometry from $\C^m$ to the single-excitation subspace.
	
	\begin{lemma}[Gaussian identity]
		\label{lem:identity}
		If $S$ is real symmetric, $S^2=I$, and $\operatorname{tr}S=0$, then
		\begin{equation}
			Q_S\Fou Q_S\Fou Q_S\phi_\alpha=-i\phi_{S\alpha}.
			\label{eq:identity}
		\end{equation}
	\end{lemma}
	\begin{proof}
		On functions $\kappa(b^\T x)e^{-x^\T Ax/2}$ with $\operatorname{Re}A>0$,
		a Gaussian integral sends
		$(A,b,\kappa)$ under $\Fou$ to
		$(A^{-1},-iA^{-1}b,\kappa\det(A)^{-1/2})$; the square root is continued
		from $A=I$. The five operations therefore give
		\begin{align}
			A:\quad I&\to I-iS\to\tfrac12(I+iS)
			\to\tfrac12(I-iS)\to I+iS\to I, \label{eq:gaussian-flow}\\
			b&\to-\tfrac i2(I+iS)b
			\to-\tfrac12(I+iS)^2b=-iSb. \nonumber
		\end{align}
		For eigenvalues $s_j=\pm1$ of $S$, the determinant factors multiply to
		$\prod_j\frac{1+is_j}{\sqrt2}
		=e^{i\pi\operatorname{tr}S/4}=1$.
	\end{proof}
	
	We use the realification and involution reductions of
	Ref.~\cite[Lemma~2.1]{nehoran2026}. For $U\in\mathrm U(d)$, put
	$H=\ket0\!\bra1\otimes U^\dagger+\ket1\!\bra0\otimes U$, also used in
	Ref.~\cite[following Proposition~5.12]{ssw2026}, and
	\begin{equation}
		S_U=I\otimes\operatorname{Re}H+
		\begin{pmatrix}0&-1\\1&0\end{pmatrix}\otimes\operatorname{Im}H.
		\label{eq:real-form}
	\end{equation}
	This is a real symmetric traceless involution on $m=4d$ coordinates.
	For $\ket{-i}=(\ket0-i\ket1)/\sqrt2$,
	\begin{equation}
		S_U\ket{-i}\ket0\ket\psi=\ket{-i}\ket1U\ket\psi.
		\label{eq:involution-action}
	\end{equation}
	Thus two label qubits and fixed input and output rotations reduce
	synthesis of $U$ to Eq.~\eqref{eq:identity}.
	
	\subsection{Periodization and a uniform error bound}
	Set $K=2^q\ge16$, $\delta=\sqrt{2\pi/K}$, and
	$I_K=\{-K/2,\ldots,K/2-1\}$. For Schwartz functions define
	\begin{equation}
		(Pf)(k)=\delta^{m/2}\sum_{\ell\in\Z^m}
		f\bigl(\delta(k+K\ell)\bigr),\qquad k\in I_K^m.
		\label{eq:periodize}
	\end{equation}
	Let $u_s=P\psi_s$ in one dimension, where
	$\psi_0(x)=\pi^{-1/4}e^{-x^2/2}$ and
	$\psi_1(x)=\sqrt2x\psi_0(x)$. The two vectors have opposite parity
	under $k\mapsto-k\pmod K$, so they are orthogonal. Define $\Phi e_j$
	as the product of normalized $u_1$ in mode $j$ and normalized $u_0$
	in the other modes. Factorization gives
	\begin{equation}
		P\phi_\alpha=\|u_1\|\,\|u_0\|^{m-1}\Phi\alpha.
		\label{eq:encoding-norm}
	\end{equation}
	Writing
	$\theta_K=\sum_{t\in\Z}e^{-\pi Kt^2/2}$ and
	$M_K=\sum_{t\in\Z}t^2e^{-\pi Kt^2/2}$, Poisson summation gives
	\begin{equation}
		\|u_0\|^2=\theta_K^2\ge1,
		\qquad
		\|u_1\|^2=\theta_K^2-2\pi K\theta_KM_K
		=1+O(Ke^{-\pi K/2}).
		\label{eq:sampled-norms}
	\end{equation}
	Consequently the prefactor in Eq.~\eqref{eq:encoding-norm} is at least
	$1/2$ once $K$ exceeds an absolute constant, independently of $m$.
	
	Let $\Fou_K$ have entries $K^{-1/2}e^{-2\pi ijk/K}$, and let $D_A$
	multiply $\ket k$ by $e^{i\delta^2k^\T Ak/2}$.
	\begin{lemma}[Discretized identity]
		\label{lem:discretization}
		Under the hypotheses of Lemma~\ref{lem:identity}, there are absolute
		constants $C,c>0$ such that, whenever $K\ge C\log(m+1)$,
		\begin{equation}
			\bigl\|D_S\Fou_K^{\otimes m}D_S\Fou_K^{\otimes m}D_S\Phi
			+i\Phi S\bigr\|_{\mathrm{op}}
			\le Cm(m+\sqrt K)e^{-cK}.
			\label{eq:discretization}
		\end{equation}
	\end{lemma}
	\begin{proof}
		Poisson summation and $\delta^2K=2\pi$ give the exact relation
		\begin{equation}
			\Fou_K^{\otimes m}Pf=P\Fou f.
			\label{eq:poisson}
		\end{equation}
		Thus only the phase multiplications contribute to the error.
		The terms with $\ell=0$ agree and all other terms differ by phases,
		so
		\begin{equation}
			\|D_APf-PQ_Af\|\le2\tau(f),\qquad
			\tau(f)=\sum_{\ell\ne0}\left(\delta^m\sum_{k\in I_K^m}
			|f(\delta(k+K\ell))|^2\right)^{1/2}.
			\label{eq:alias}
		\end{equation}
		Let $f_0,\ldots,f_5$ be the successive continuous states in
		Eq.~\eqref{eq:identity}, starting at $\phi_\alpha$ with $\|\alpha\|=1$.
		Telescoping through the unitary discrete gates bounds the error on
		$P\phi_\alpha$ by $2[\tau(f_0)+\tau(f_2)+\tau(f_4)]$.
		
		These three functions have unit norm and the form
		$\kappa(b^\T x)e^{-x^\T Ax/2}$, with $\operatorname{Re}A=\gamma I$ for
		$\gamma=1,1/2,1$. Their normalization and Cauchy--Schwarz imply
		\[
		|f(x)|^2\le2\gamma(\gamma/\pi)^{m/2}|x|^2e^{-\gamma|x|^2},
		\]
		uniformly in $\alpha$. For $j=0,2$, set
		$G_j(t)=\delta\sum_{k\in I_K}y^je^{-\gamma y^2}$ with
		$y=\delta(k+Kt)$, and
		$a_j=(\gamma/\pi)^{1/4}\sqrt{G_j(0)}$,
		$s_j=(\gamma/\pi)^{1/4}\sum_{t\ne0}\sqrt{G_j(t)}$.
		Writing $|x|^2=\sum_ax_a^2$ and applying
		$\sqrt{\sum_ay_a}\le\sum_a\sqrt{y_a}$ gives
		\begin{equation}
			\tau(f)\le\sqrt{2\gamma}\,m
			\left[s_2(a_0+s_0)^{m-1}
			+a_2\bigl((a_0+s_0)^{m-1}-a_0^{m-1}\bigr)\right].
			\label{eq:factorized-tail}
		\end{equation}
		Poisson summation of the full one-dimensional grid bounds
		$a_0\le1+Ce^{-cK}$ and $a_2\le(2\gamma)^{-1/2}+Ce^{-cK}$.
		For a noncentral cell, all $|y|$ are at least
		$L_t=\sqrt{\pi K/2}(2|t|-1)$. The integral estimates
		\begin{align*}
			\int_a^\infty e^{-\gamma y^2}\,dy
			&\le\frac{e^{-\gamma a^2}}{2\gamma a},\\
			\int_a^\infty y^2e^{-\gamma y^2}\,dy
			&\le\left(\frac a{2\gamma}+\frac1{4\gamma^2a}\right)e^{-\gamma a^2}
		\end{align*}
		and bounding each decreasing grid sum by its first term plus the
		corresponding integral give
		\begin{align}
			G_0(t)&\le C(\delta+L_t^{-1})e^{-\gamma L_t^2},\nonumber\\
			G_2(t)&\le C(\delta L_t^2+L_t+L_t^{-1})e^{-\gamma L_t^2}.
			\label{eq:tail-cells}
		\end{align}
		Taking square roots and summing over $t\ne0$ gives
		$s_0=O(\eta)$ and $s_2=O(\sqrt K\,\eta)$ for
		$\eta=e^{-cK}$. Constants are uniform for $1/2\le\gamma\le1$.
		If $K\ge C\log(m+1)$ with sufficiently large $C$, then
		\begin{align}
			(a_0+s_0)^{m-1}&=O(1),\nonumber\\
			(a_0+s_0)^{m-1}-a_0^{m-1}&=O(m\eta),\nonumber\\
			\tau(f)&=O\!\left(m(m+\sqrt K)\eta\right).
			\label{eq:uniform-gaussian-tail}
		\end{align}
		Apply this to the three continuous states and divide by the common
		prefactor in Eq.~\eqref{eq:encoding-norm}. The estimate holds for every
		unit vector $\alpha$, proving the operator-norm bound.
	\end{proof}
	
	\subsection{Circuit implementation and width}
	\begin{lemma}[Small Fourier transforms]
		\label{lem:small-fourier}
		For $K=2^q$, $\Fou_K$ has an exact clean constant-depth implementation
		with one- and two-qubit gates and fanout on $O(K^4)$ qubits.
	\end{lemma}
	This is Nehoran--Yuen's exact Fourier circuit
	\cite[Sec.~5.3.1, Algorithm~6]{nehoran2026}, based on Fourier-state
	preparation~\cite[Lemma~4.13]{hoyer2005} and exact phase
	decoding~\cite[Sec.~4.1]{takahashi2016}, and including the final
	modular sum and ancillary cleanup. Their Fourier convention
	has the opposite sign; reversing the circuit gives $\Fou_K$ with the
	same resources.
	
	\begin{proof}[Proof of Lemma~\ref{lem:seed}]
		Choose the smallest power of two $K$ satisfying the hypothesis of
		Lemma~\ref{lem:discretization} and making its bound at most
		$\varepsilon$. Then $K=O(\log(d/\varepsilon))$ for $m=4d$.
		The input rotations prepare $\ket{-i}\ket0$ on the label qubits.
		An encoder maps $\ket j\ket0$ to $\ket0\Phi e_j$ as follows.
		Convert the label to one-hot form by
		Eq.~\eqref{eq:onehot-encoding}, using $O(m)$ qubits.
		
		In each mode prepare $u_0/\|u_0\|$ and, controlled on its indicator,
		exchange it with $u_1/\|u_1\|$ by the Householder reflection with vector
		$v=(u_0/\|u_0\|-u_1/\|u_1\|)/\sqrt2$. If
		$P_v\ket0=\ket{v,0}$, apply $P_v^\dagger$ first, reflect about the zero
		state of the mode and \emph{all} its preparation work, and then apply
		$P_v$. The resulting operator is
		\begin{equation}
			P_v(I-2\ket0\!\bra0)P_v^\dagger
			=I-2\ket{v,0}\!\bra{v,0}.
			\label{eq:ordinary-householder}
		\end{equation}
		Thus the reflection preserves the zero-work subspace. The preparation uses
		$s=O(q^2K)$ qubits by Lemma~\ref{lem:stateprep}. Compiling the
		reflection on these qubits costs
		$O(s\log(s+1))=O(q^3K)$ by Lemma~\ref{lem:toffoli}; controlling the
		other gates adds only $O(s)$ work. This is within the $O(K^4)$ bound.
		The states $u_0,u_1$ have opposite eigenvalues under modular negation.
		We erase the indicator by adapting the phase-kickback construction of
		Ref.~\cite[Algorithm~3]{nehoran2026}:
		apply a Hadamard to the indicator $b$, controlled modular negation
		$k\mapsto-k\pmod K$ on the mode, and another Hadamard to $b$.
		For each initial indicator value $b$, the mode has negation eigenvalue
		$(-1)^b$, so this maps $\ket b\ket{u_b/\|u_b\|}$ to
		$\ket0\ket{u_b/\|u_b\|}$. Controlled modular negation is a classical
		permutation on $q+1$ bits. Parallel truth-table tests, compiled using
		Lemma~\ref{lem:toffoli}, implement it exactly in constant depth with
		clean work and $O(Kq\log(q+1))$ qubits.
		
		Apply the five discrete layers, reverse the encoder, and use the
		fixed output rotations and phase $i$ specified by
		Eq.~\eqref{eq:involution-action}. In two's complement,
		$k_a=\sum_i c_i k_{a,i}$ with
		$c_i=2^i$ except $c_{q-1}=-2^{q-1}$, so
		\begin{equation}
			\frac{\delta^2}{2}k^\T Ak
			=\frac{\delta^2}{2}\sum_{a,b,i,i'}
			A_{ab}c_ic_{i'}k_{a,i}k_{b,i'}.
			\label{eq:bitwise-quadratic}
		\end{equation}
		After collecting terms this is a product of single-qubit and
		controlled-phase gates on $mq$ bits. Fanout provides a separate
		pair of copies for each interaction. Apply all phases in parallel and
		reverse the fanouts, at constant depth and width $O(m^2q^2)$.
		This replaces the product computation in
		Ref.~\cite[Algorithm~4]{nehoran2026}, whose phase-layer width is
		$O(m^2K^2q\log q)$. The Fourier layers and encoder require
		$O(mK^4)$ qubits. Combining the two bounds gives
		\[
		O\!\left(d^2[\log\log(d/\varepsilon)]^2
		+d\log^4(d/\varepsilon)\right)
		\]
		width. Write $\mathcal D_S=D_S\Fou_K^{\otimes m}D_S\Fou_K^{\otimes m}D_S$.
		Since $\Phi^\dagger\Phi=I$, the decoded block satisfies
		\[
		\bigl\|i\Phi^\dagger\mathcal D_S\Phi-S\bigr\|_{\mathrm{op}}
		\le\bigl\|\mathcal D_S\Phi+i\Phi S\bigr\|_{\mathrm{op}}
		\le\varepsilon.
		\]
		The fixed label preparations and output rotations do not increase
		this error.
	\end{proof}

	\section{Exact state preparation}
	\label{app:stateprep}
	
	\begin{lemma}[Generalized Toffoli gates]
		\label{lem:toffoli}
		A generalized Toffoli with $k$ controls has an exact clean
		constant-depth circuit of one- and two-qubit gates and fanout on
		$O(k\log\log(k+3))$ qubits.
	\end{lemma}
	Takahashi and Tani's refinement of their exact OR construction has
	constant depth and size $O(k\log\log(k+3))$
	\cite[Sec.~3.2]{takahashi2016}. Their size convention sums all gate
	arities and hence bounds the workspace. Negating controls
	and output gives AND; compute--copy--uncompute gives the clean XOR
	into an arbitrary target with the same resource orders.
	
	\begin{proof}[Proof of Lemma~\ref{lem:stateprep}]
		Put $N=2^q$ and write the target as
		$\sum_{j<N}\beta_je^{i\varphi_j}\ket j$, with $\beta_j\ge0$.
		The first-one construction of Ref.~\cite[Theorem~4.2]{ssw2026}
		uses the predicates
		$g_k(x)=[x_0=\cdots=x_{k-1}=0]$ for $1\le k<N$.
		They can all be computed cleanly in constant depth on $O(qN)$
		qubits with generalized Toffoli gates. Write $k=\sum_{\ell=0}^{q-1}k_\ell2^\ell$.
		For $0\le\ell<q$ and $0\le a<2^{q-\ell-1}$, the dyadic blocks
		$B_{\ell,a}=\{a2^{\ell+1},\ldots,a2^{\ell+1}+2^\ell-1\}$ satisfy
		\begin{equation}
			\{0,\ldots,k-1\}
			=\bigcup_{\ell:k_\ell=1}B_{\ell,\lfloor k/2^{\ell+1}\rfloor}.
			\label{eq:dyadic-prefix}
		\end{equation}
		Each position belongs to at most $q$ blocks, and $B_{\ell,a}$ is used
		by $2^\ell$ prefixes. Copy the input bits, compute all block-zero
		indicators, copy those indicators to their prefixes, and compute
		each $g_k$ as a conjunction of at most $q$ indicators. Reversing
		the intermediate computations leaves only the input and flags.
		Both copying stages use $O(qN)$ qubits.
		
		Let $J$ be the last index with $\beta_J>0$,
		$t_j=\sum_{k\ge j}\beta_k^2$, and choose
		$\sin\theta_j=\beta_j/\sqrt{t_j}$ for $j\le J$, with
		$0\le\theta_j\le\pi/2$ and $\theta_j=0$ otherwise. Since
		$\sin\theta_J=1$, the all-zero branch vanishes. The product rotations give
		\begin{align*}
			\bigotimes_{j<N}(\cos\theta_j\ket0+\sin\theta_j\ket1)
			&=\sum_{j\le J}\beta_j\ket{0^j1}
			\bigotimes_{k>j}(\cos\theta_k\ket0+\sin\theta_k\ket1),\\
			\sin\theta_j\prod_{k<j}\cos\theta_k
			&=\frac{\beta_j}{\sqrt{t_j}}
			\prod_{k<j}\sqrt{\frac{t_{k+1}}{t_k}}=\beta_j.
		\end{align*}
		Compute the prefix flags. On the summand with first one at $j$,
		$g_k=0$ exactly when $k>j$, so controlled inverse rotations return
		all later qubits to zero. Changing bits after the first one leaves every prefix predicate
		unchanged, so the same prefix circuit clears the flags. Applying phase $e^{i\varphi_j}$ to qubit
		$j$ gives the desired one-hot state. Parities write its binary
		label into a fresh $q$-qubit output, and equality tests on copies
		of that label clear the one-hot register. All temporary copies are
		uncomputed before their source registers change.
		
		The depth is constant and the width is $O(qN)$. A layer whose
		generalized Toffolis have $W=O(qN)$ controls in total expands, by
		Lemma~\ref{lem:toffoli}, to $O(W\log(W+1))=O(q^2N)$ qubits.
		The added compilation work is clean and reusable, proving both
		bounds.
	\end{proof}
	
	\section{Coherent selection and the depth--width trade-off}
	\label{app:tradeoff}
	
	We use binary-to-one-hot conversion as in
	Refs.~\cite[Lemma~2.16]{ssw2026} and~\cite[Algorithm~2]{nehoran2026},
	splitting the label to reduce the workspace. An $r$-bit label has a clean encoding
	\begin{equation}
		\ket c\ket{0^{2^r}}\longmapsto\ket{0^r}\ket{e_c}
		\label{eq:onehot-encoding}
	\end{equation}
	by a constant-depth fanout circuit of width $O(2^r)$.
	Here $e_c$ is the bit string with a single $1$ at position $c$.
	Split $c=(u,v)$ into two halves, compute all equality indicators
	for each half, and erase the original bits by parities of those
	indicators. Products of the half indicators give every $[c=c']$.
	The half indicators are then erased by summing the appropriate full
	indicators. Temporary copies are uncomputed before their sources
	change. The generalized-Toffoli compilation for the half tests costs
	$O(r\log(r+1)2^{\lceil r/2\rceil})=O(2^r)$ qubits, and combining and
	clearing the indicators costs $O(2^r)$. The case $r=0$ is immediate.
	
	\begin{lemma}[Parallel coherent selection]
		\label{lem:parallel-selection}
		Let $C_c$ be specified coherent circuits, each on $k$ data qubits
		with its own work register, with depths $h_c$, widths $w_c$, and
		blocks $B_c$. With one- and two-qubit gates, fanout, and generalized Toffoli,
		there is a circuit $S$ of depth $O(1+\max_c h_c)$ and width
		$O(\sum_c w_c+(k+1)2^r)$ such that
		\begin{equation}
			J^\dagger S J=\sum_{c\in\{0,1\}^r}\ket c\bra c\otimes B_c.
			\label{eq:selected-block}
		\end{equation}
		Here $J$ appends zero-initialized work qubits to the selector and data registers.
		The circuits $C_c$ may leave nonzero work.
	\end{lemma}
	
	\begin{proof}
		Compute $o_{c'}=[c=c']$ while retaining the binary selector, using
		Eq.~\eqref{eq:onehot-encoding} on a computational-basis copy of it.
		For each data bit $x_j$, parallel Toffoli gates write $o_cx_j$ into
		bit $j$ of a separate data register $\reg S_c$. Uncompute the
		temporary copies, then erase $x_j$ by the parity of bit $j$ across
		these registers. This routes the data into the selected register
		and leaves the original data register zero.
		Run each $C_c$ controlled by its own indicator, in parallel, then
		reverse the routing and indicator computation. All unselected
		circuits act as identity. The inverse routing only requires the other
		data registers to be zero; the selected register can be entangled with its
		work and an arbitrary reference. Projecting all circuit work onto
		zero therefore gives Eq.~\eqref{eq:selected-block} by linearity.
		The controlled-circuit construction of Sec.~\ref{sec:repair} costs
		$O(h_c+1)$ depth and $O(w_c)$ width, and routing costs
		$O((k+1)2^r)$ qubits at constant depth.
	\end{proof}
	
	\begin{lemma}[Exact multiplexers]
		\label{lem:mux}
		For $W_c\in\mathrm U(2^k)$, the multiplexer
		\begin{equation}
			M=\sum_{c\in\{0,1\}^r}\ket c\bra c\otimes W_c
			\label{eq:multiplexer}
		\end{equation}
		has an exact clean constant-depth fanout circuit on
		$O((k+1)^2 2^r4^k)$ qubits.
	\end{lemma}
	
	\begin{proof}
		Use the routing above and implement each $(k+1)$-qubit controlled
		$W_c$ by Theorem~\ref{thm:main}, each with its own clean work
		register. These circuits act on disjoint registers, so their depth
		is constant and their
		total width is $O((k+1)^2 2^r4^k)$. Reversing the routing completes
		the multiplexer. Any global phases of these complete circuits
		multiply to a single phase independent of the selector.
	\end{proof}
	
	\begin{proof}[Proof of Theorem~\ref{thm:tradeoff}]
		Apply the cosine--sine decomposition recursively through the first
		$r$ qubits~\cite{paigewei1994,mottonen2004,shende2006}. With $k=n-r$,
		it gives
		\begin{equation}
			U=M_1R_1M_2R_2\cdots R_{2^r-1}M_{2^r},
			\label{eq:csd}
		\end{equation}
		where each $M_i$ is an $r$-controlled multiplexer on $k$ qubits.
		Each $R_i$ rotates one of the first $r$ qubits, conditioned on the
		others, and can be viewed as an $(r-1)$-controlled multiplexer on
		$k+1$ qubits. Lemma~\ref{lem:mux} implements every factor at constant
		depth with width $O((k+1)^2 4^n/2^r)$. Reuse the work between the
		$2^{r+1}-1$ factors and take $r=\lfloor\log_2h\rfloor$.
		Replacing fanout after completing the coherent circuit gives the
		adaptive bound.
	\end{proof}
	
	At smaller widths, the ordinary-circuit bound
	$O(nd+d^2/w)$ of Sun et al.~\cite{sun2023} fills the remaining range:
	it is optimal for $n\le w=O(d/n)$ and within a factor $O(n)$ of the
	counting bound for $d/n\lesssim w\lesssim d$. Thus the adaptive
	depth--width frontier is determined up to factors polynomial in
	$n=\log_2d$, including this transition interval.
	
	\begin{corollary}[Logarithmic blocks]
		\label{cor:log}
		A constant number of layers of disjoint $O(\log N)$-qubit unitaries on $N$ qubits
		has an exact clean constant-depth fanout circuit on
		$\operatorname{poly}(N)$ qubits, and an adaptive implementation with
		the same resource order. This includes depth-$O(\log N)$
		nearest-neighbor circuits on a line of $N$ qubits.
	\end{corollary}
	
	\begin{proof}
		Apply Theorem~\ref{thm:main} to each block. For a line circuit of
		depth $D$, number the layers $1,\ldots,D$. The case $N\leq2D+4$
		is already a single block.
		Otherwise mark every $(2D+4)$th site as a cut. Gates whose time is
		at most the distance of both sites from the nearest cut form a set
		closed under predecessors, so they can be applied first. They act
		within disjoint intervals between cuts. The remaining gates lie
		within distance $D+1$ of a cut and form another layer of disjoint
		$O(D)$-site blocks. Apply the first assertion to these two layers.
	\end{proof}

	\section{Fanout by measurement and parity feedforward}
	\label{app:fanout}
	
	We give the seven-layer construction used in
	Sec.~\ref{sec:adaptive}, including its action on each measurement
	branch.
	
	Let $\reg c$ be the control, $\reg t_1,\ldots,\reg t_q$ the targets, and
	$\reg R$ a reference. An arbitrary input has the form
	$\sum_{x=0}^1\ket x_{\reg c}\ket{\phi_x}$, where the states
	$\ket{\phi_x}$ of targets and reference may be unnormalized and
	nonorthogonal. The gadget uses cat qubits $\reg a_0,\ldots,\reg a_q$
	and probes $\reg p_1,\ldots,\reg p_q$, all starting in $\ket0$.
	The first three layers of Fig.~\ref{fig:fanout} apply Hadamards to
	all cat qubits, followed by the CNOT layers
	$\reg a_{j-1}\to\reg p_j$ and $\reg a_j\to\reg p_j$ for all $j$.
	The fourth applies the CNOTs $\reg c\to\reg a_0$ and
	$\reg a_j\to\reg t_j$ in parallel. We then apply Hadamards to
	$\reg a_1,\ldots,\reg a_q$, measure all cat qubits and probes, and
	apply the corrections below. Write $s_j$ for the probe outcomes,
	$\mu$ for the outcome
	on $\reg a_0$, $o_j$ for the outcome on $\reg a_j$, and
	$b_j=s_1\oplus\cdots\oplus s_j$. The last layer applies
	\begin{equation}
		Z_{\reg c}^{\,o_1\oplus\cdots\oplus o_q}\,
		\prod_{j=1}^qX_{\reg t_j}^{\,b_j\oplus\mu}
		\label{eq:fanout-correction}
	\end{equation}
	and resets each measured qubit by $X$ raised to its recorded bit.
	
	\begin{figure}[htbp]
		\centering
		\includegraphics[width=0.6\textwidth]{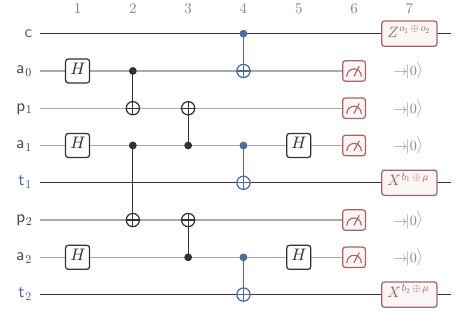}
		\caption{Seven-layer implementation of $F_2$ by measurement and parity
			feedforward, with all ancillas initialized and reset to $\ket0$.
			All measurements occur in layer~6; the corrections of
			Eq.~\eqref{eq:fanout-correction} and resets occur in layer~7.
			For $q$ targets, the construction uses $2q+1$ ancillas.}
		\label{fig:fanout}
	\end{figure}
	
	For a fixed probe record, the unnormalized cat state is
	$2^{-(q+1)/2}\sum_y\ket y\bigotimes_j\ket{y\oplus b_j}$.
	The outcome $\mu$ forces $y=x\oplus\mu$, and the final Hadamards
	contribute $2^{-q/2}(-1)^{\sum_j o_j(x\oplus\mu\oplus b_j)}$.
	After Eq.~\eqref{eq:fanout-correction}, the branch operator is therefore
	\begin{equation}
		K^{\rm corr}_{s,\mu,o}
		=2^{-(2q+1)/2}(-1)^{\sum_j o_j(\mu\oplus b_j)}F_q.
		\label{eq:fanout-branch}
	\end{equation}
	Every record has probability $2^{-(2q+1)}$, independent of the input,
	and implements $F_q$ even with an entangled reference.
	
	In a layer of width $w$, the fanout gates have disjoint supports and
	use at most $2w$ extra qubits. These are reused after reset, giving
	Eq.~\eqref{eq:measurement-compiler}.
	
	\section{An ordinary depth--width trade-off}
	\label{app:ordinary-tradeoff}
	
	The Gaussian approximate circuit can also be compiled directly for
	limited workspace.
	This improves on replacing every fanout in the constant-depth circuit
	by a tree. Throughout this section only one- and two-qubit gates are
	allowed, and width includes all data and work registers.
	
	\Needspace{10\baselineskip}
	\begin{theorem}[Ordinary depth--width trade-off]
		\label{thm:ordinary-tradeoff}
		For an absolute constant $C$, every $U\in\mathrm U(d)$, $d=2^n$,
		and every $Cdn\le w\le Cn^2d^2$ admit an exact clean circuit
		of width $O(w)$ and depth
		\begin{equation}
			O\!\left(n+\frac{d^2\log^2(n+2)}w\right).
			\label{eq:ordinary-width-bound}
		\end{equation}
		For generic targets the lower bound is $\Omega(n+d^2/w)$, so the
		upper bound is within a factor $O(\log^2(n+2))$ in this width range.
	\end{theorem}
	Larger widths use the construction at the upper endpoint, with the
	additional qubits idle.
	
	\begin{lemma}[Controlling an ordinary circuit]
		\label{lem:ordinary-control}
		A depth-$t$, width-$W$ ordinary circuit can be controlled by one qubit
		in depth $O(t+\log(W+2))$ and width $O(W)$, with clean additional
		work. A reflection about the zero state of any specified set of its
		wires has depth $O(\log(W+2))$ and width $O(W)$.
	\end{lemma}
	\begin{proof}
		A balanced CNOT tree distributes the external control to one separate
		copy per physical wire. Each gate uses a designated copy from its
		support, so disjoint gates use distinct controls. Controlled local
		gates have constant-size exact decompositions~\cite{barenco1995}.
		The copies remain controls throughout all $t$ layers and are cleared
		by one inverse tree at the end. For a reflection, a balanced Toffoli
		tree computes the conjunction of the negated tested bits, a $Z$
		changes its sign, and the inverse tree clears the added work.
		This works on arbitrary states of the tested wires.
	\end{proof}
	
	\begin{lemma}[Quadratic phases with limited workspace]
		\label{lem:ordinary-quadratic}
		A diagonal phase that is a quadratic function of $M$ bits has an
		exact clean implementation of width $O(W)$ and depth
		\begin{equation}
			O\!\left(\log(M+2)+1+\frac{M^2}{W}\right),\qquad W\ge M.
			\label{eq:quadratic-schedule}
		\end{equation}
	\end{lemma}
	\begin{proof}
		After collecting terms, the two-bit phases are edges of a simple
		graph on $M$ vertices. Give each bit
		$L=\min\{M,\lfloor W/M\rfloor\}$ representatives, including the
		original, using parallel balanced CNOT trees. At each vertex,
		distribute its incident edges among its representatives, each
		receiving at most $\lceil(M-1)/L\rceil$ edges. An original edge now
		joins its assigned representative at each endpoint. The resulting
		graph has maximum degree $O(M/L)$, so greedy edge coloring gives
		$O(1+M/L)$ matchings. Apply these layers and the single-bit phases,
		then reverse the copy trees. The phases preserve computational-basis
		values, so all copies are erased exactly. The width is $O(ML)$ and
		the depth $O(\log(L+1)+1+M/L)$, as claimed.
	\end{proof}
	
	\begin{proof}[Proof of Theorem~\ref{thm:ordinary-tradeoff}]
		At approximation error $1/d$, the discretization of
		Supplementary Material~\ref{app:seed} permits
		\begin{equation}
			m=4d,\qquad K=2^q=\Theta(n),\qquad q=O(\log(n+2)),
			\label{eq:ordinary-grid}
		\end{equation}
		with an absolute, target-independent constant in $K$.
		Use the same encoder as there, with ordinary implementations of its
		components. Their costs, with all modes processed in parallel, are
		\begin{center}
			\begin{tabular}{@{}lll@{}}
				\toprule
				Operation & Width & Depth \\
				\midrule
				Binary/one-hot conversion & $O(m\log m)$ & $O(\log m)$ \\
				Mode preparation and controlled exchange & $O(mK)$ & $O(q)$ \\
				Parity readout by modular negation & $O(mq)$ & $O(q)$ \\
				Ordinary Fourier transforms & $O(mq)$ & $O(q^2)$ \\
				\bottomrule
			\end{tabular}
		\end{center}
		For the conversion in Eq.~\eqref{eq:onehot-encoding}, separate copies
		and balanced AND/parity trees give the first row.
		For the second, Yuan--Zhang state preparation on $q$ logical qubits
		with $a=\Theta(K)$ ancillas has depth
		$O(q+2^q/(q+a))=O(q)$ and clean work
		\cite[Theorem~2]{yuanzhang2023}. Lemma~\ref{lem:ordinary-control}
		implements the controlled Householder reflection in
		Eq.~\eqref{eq:ordinary-householder} with the same resource orders. Modular negation
		uses two's complement and a ripple-carry increment. For the last
		row, use the standard exact $q$-qubit Fourier circuit; the much wider
		constant-depth Fourier circuit is unnecessary. Thus the encoder and
		its inverse have
		\begin{equation}
			\text{width }O(mK+m\log m)=O(dn),\qquad
			\text{depth }O(n).
			\label{eq:ordinary-encoder}
		\end{equation}
		
		Each quadratic phase acts on $M=mq=O(d\log(n+2))$ bits.
		Allocate a constant fraction of the $O(w)$ workspace to its copies
		and apply Lemma~\ref{lem:ordinary-quadratic}. The remaining live
		registers fit because $w\ge Cdn$. Combining the three phases with
		the table gives a width-$O(w)$ approximate circuit of depth
		\begin{equation}
			O\!\left(n+\frac{d^2\log^2(n+2)}w\right).
			\label{eq:ordinary-seed-depth}
		\end{equation}
		Only the scheduling has changed, so its block-error guarantee is
		unchanged.
		
		The residual has $2n$ logical qubits. The same Yuan--Zhang theorem,
		now with $\Theta(w)$ ancillas, prepares it exactly and cleanly in
		width $O(w)$ and depth
		\begin{equation}
			O\!\left(n+\frac{d^2}w\right).
			\label{eq:ordinary-residual-depth}
		\end{equation}
		Apply the residual-repair circuit using
		Lemma~\ref{lem:ordinary-control} for its constant number of controls
		and reflections. Their total additional depth is $O(\log w)=O(n)$
		and their clean additional workspace is $O(w)$. This proves the
		upper bound, with the same exact block identity and cleanup as in
		Theorem~\ref{thm:repair}. The lower bound follows by combining
		parameter counting with the ordinary causal-cone bound.
	\end{proof}

	\section{Exact synthesis from constant-accuracy approximations}
	\label{app:constant-accuracy}
	
	Suppose every $T\in\mathrm U(d)$, $d=2^n$, has a specified coherent
	circuit of depth $h$ and width $w$ whose block $B_T$
	satisfies $\opnorm{B_T-T}\leq1/3$. We use one- and two-qubit gates,
	fanout, and generalized Toffoli.
	
	Every square matrix $M$ with $\opnorm M\le1$ is the average of
	two unitaries. For a full SVD
	$M=W_1\operatorname{diag}(s_i)W_2^\dagger$, $0\leq s_i\leq1$, set
	\begin{equation}
		V_\pm=W_1\operatorname{diag}\!\left(s_i\pm i\sqrt{1-s_i^2}\right)W_2^\dagger,
		\qquad M=(V_++V_-)/2.
		\label{eq:constant-two-unitaries}
	\end{equation}
	Fix $U$, let $L=\lceil\log_3(2d)\rceil$, and start with
	$B_0=B_U$, $E_1=U-B_0$. At round $j$, the matrix $3^jE_j$ has
	operator norm at most one. Decompose it as above and approximate
	its two unitaries by $B_{j,\pm}$:
	\begin{align}
		E_{j+1}&=E_j-\frac{3^{-j}}2(B_{j,+}+B_{j,-}),
		&\opnorm{E_{j+1}}&\leq3^{-(j+1)},                        \label{eq:constant-recurrence}\\
		U&=B_0+\sum_{j=1}^{L-1}\frac{3^{-j}}2(B_{j,+}+B_{j,-})+E_L,
		&\opnorm{E_L}&\leq\frac1{2d}.
		\label{eq:constant-expansion}
	\end{align}
	The resulting $2L-1$ approximate circuits are specified classically.
	
	For $E_L\neq0$, suppose the residual state in
	Eq.~\eqref{eq:residual-state}, with $E=E_L$, has an exact clean
	preparation of depth $h_E$ and width $w_E$.
	Equation~\eqref{eq:residual-block} gives the block $E_L/\nu$, and
	\begin{align}
		\nu&=\sqrt d\,\fnorm{E_L}\leq d\,\opnorm{E_L}\leq\tfrac12,
		&\lambda&=1+\sum_{j=1}^{L-1}3^{-j}+\nu\leq2.
		\label{eq:constant-normalization}
	\end{align}
	When $E_L=0$, omit the residual branch and set $\nu=h_E=w_E=0$. Prepare
	the selector with probabilities $1/\lambda$, $3^{-j}/(2\lambda)$
	for each sign at round $j$, and $\nu/\lambda$ for the residual
	branch. Pad to a power of two with zero-amplitude branches and
	identity circuits.
	Lemma~\ref{lem:parallel-selection} implements the selection in
	parallel. After unpreparing the selector, the block is $U/\lambda$.
	Rotating a zero-initialized flag by $R_y(\theta)$ with
	$\cos(\theta/2)=\lambda/2$ gives $U/2$, to which the exact identity
	\eqref{eq:cubic} applies.
	
	The selector state has $O(L)$ nonzero amplitudes and needs constant depth and
	$O(L\log(L+1))$ qubits by Lemma~\ref{lem:stateprep}. The selection
	lemma, $L=O(n)$, and $w\geq n$ therefore give
	\begin{equation}
		H=O(h+h_E+1),\qquad W=O(Lw+w_E)=O(nw+w_E).
		\label{eq:constant-resources}
	\end{equation}
	Amplification uses each specified approximate circuit twice forward and once
	backward. Compiling generalized Toffoli gates, including the
	all-work reflections, gives a fanout circuit of width
	$O(W\log(W+1))$, and the fanout gadget gives the adaptive circuit.
	If both the approximate circuits and the residual preparation use ordinary
	one- and two-qubit gates, Lemma~\ref{lem:ordinary-control} and tree
	implementations of routing and selector preparation instead give
	depth $O(h+h_E+\log(W+1))$ and width $O(W)$.

	\section{An independent proof of exact constant-depth universality}
	\label{sm:independent}
	
	The following finite-dimensional construction gives exact universality
	with width $2^{O(2^n)}$. Its synthesis ingredients are the elementary
	binary-to-one-hot conversion \eqref{eq:onehot-encoding}, the standalone
	state-preparation Lemma~\ref{lem:stateprep}, and classical reversible
	lookup. It does not invoke the approximate synthesis theorem or the
	exact multiplexer lemma.
	
	Let $N=2^D$ and index its basis by $x\in\{-1,+1\}^D$. Define the
	fixed real isometry $F_{x,j}=x_j/\sqrt N$, with $F^{\mathsf T}F=I_D$.
	For every real matrix $M\in\mathbb R^{D\times D}$ with
	$\opnorm M\le1$, we will show that
	\begin{equation}
		F^{\mathsf T}\Bigl(\sum_{k=1}^K\lambda_kP_k\Bigr)F=M/2,
		\qquad K\leq D^2+1,\quad \lambda_k\geq0,\quad\sum_k\lambda_k=1,
		\label{sm:eq:permutation-dilation}
	\end{equation}
	where the $P_k$ are permutations of the $N$ basis states.
	Let $M^{\circ r}$ denote the entrywise $r$th power of $M$, for
	$r\ge1$. For $J_r\ket j=\ket j^{\otimes r}$, the identity
	$M^{\circ r}=J_r^{\mathsf T}M^{\otimes r}J_r$ gives
	$\opnorm{M^{\circ r}}\leq1$. Following the sine transformation in
	Krivine's argument~\cite{krivine1979}, define
	\begin{equation}
		C_{ij}=\sin\!\left(\frac{\pi}{4}M_{ij}\right),\qquad
		\opnorm C\leq\sum_{a\geq0}\frac{(\pi/4)^{2a+1}}{(2a+1)!}
		=\sinh(\pi/4)<1.
		\label{sm:eq:sine-bound}
	\end{equation}
	Take centered jointly Gaussian $g,g'$ with covariance
	\[
	\operatorname{Cov}\!\begin{pmatrix}g\\g'\end{pmatrix}
	=\begin{pmatrix}I&C^{\mathsf T}\\ C&I\end{pmatrix}>0.
	\]
	Their sign vectors $X=\operatorname{sgn}g$ and $Y=\operatorname{sgn}g'$
	each have the uniform distribution on the entire
	sign cube because the coordinates within each marginal are
	independent. The Gaussian sign identity~\cite[Lemma~2.1]{friedland2018} gives
	\begin{align}
		\mathbb E[Y_iX_j]
		&=\frac2\pi\arcsin(C_{ij})=M_{ij}/2,
		&T_{y,x}&=N\Pr(Y=y,X=x),                                 \label{sm:eq:arcsine}\\
		\sum_yT_{y,x}&=\sum_xT_{y,x}=1,
		&(F^{\mathsf T}TF)_{ij}&=\mathbb E[Y_iX_j]=M_{ij}/2.
		\label{sm:eq:stochastic-block}
	\end{align}
	The arcsine takes its principal value since $|\pi M_{ij}/4|\leq\pi/4$.
	The doubly stochastic matrix $T$ is a convex combination of permutations
	by the Birkhoff--von Neumann theorem. Apply Carath\'eodory's theorem
	\emph{after compression}: the matrices $F^{\mathsf T}PF$ lie in a real
	space of dimension $D^2$, so at most $D^2+1$ are required. This proves
	Eq.~\eqref{sm:eq:permutation-dilation}. The Gaussian distribution
	enters only this existence argument.
	
	For a power-of-two $D$, encode $\ket j$ as $\ket{e_j}$ using
	Eq.~\eqref{eq:onehot-encoding}, then apply $H^{\otimes D}$. The
	amplitude at a bit string $z$ is $(-1)^{z_j}/\sqrt N$, which implements
	$F$ cleanly. Pad the coefficients to a power-of-two number $K=O(D^2)$
	and prepare $\ket\lambda=\sum_k\sqrt{\lambda_k}\ket k$ by
	Lemma~\ref{lem:stateprep}, using $O(D^2\log^2(D+1))$ qubits.
	
	The required selection is the classical bijection
	$\ket k\ket x\mapsto\ket k\ket{\pi_k(x)}$ on
	$m=D+\log_2K$ bits, with identity permutations on padded labels.
	With fanout and generalized Toffoli gates, any $m$-bit Boolean map
	has a clean constant-depth reversible oracle on $O(m2^m)$ qubits:
	compute all tests $[x=a]$ on fanout copies,
	form the output bits by parities, and uncompute. For a bijection $f$,
	the oracles for $f$ and $f^{-1}$ give
	$\ket x\ket0\mapsto\ket x\ket{f(x)}\mapsto\ket0\ket{f(x)}$;
	a swap restores the data location~\cite[footnote~3]{rosenthal2026}. Thus the selection operation is
	implemented at width $O(D^3 2^D)$ in the generalized-Toffoli model.
	
	Prepare $\ket\lambda$, apply the selected permutation, undo the state preparation,
	and conjugate this circuit by the encoder for $F$. With every work
	qubit included in the projection, the resulting unitary has block
	\begin{equation}
		(I\otimes\bra0)V_M(I\otimes\ket0)
		=F^{\mathsf T}\Bigl(\sum_k\lambda_kP_k\Bigr)F=M/2.
		\label{sm:eq:implemented-block}
	\end{equation}
	The permutations need not preserve the encoded subspace.
	For $U=U_{\mathrm R}+iU_{\mathrm I}\in\mathrm U(d)$, use $D=2d$ and
	\begin{equation}
		R_U=\begin{pmatrix}U_{\mathrm R}&-U_{\mathrm I}\\ U_{\mathrm I}&U_{\mathrm R}\end{pmatrix},\qquad
		\ket{-i}=(\ket0-i\ket1)/\sqrt2,\qquad
		R_U(\ket{-i}\otimes\ket\psi)=\ket{-i}\otimes U\ket\psi.
		\label{sm:eq:realification}
	\end{equation}
	Apply the construction to the real orthogonal $R_U$ and prepare and
	unprepare $\ket{-i}$. The block becomes $U/2$.
	The algebraic amplification identity \eqref{eq:cubic} therefore gives
	an exact clean implementation on every input and reference.
	
	There are a constant number of stages. Permutation selection dominates the width,
	and compiling generalized Toffoli gates, including all-work reflections,
	only multiplies it by $O(D)$. Hence both the fanout and adaptive
	implementations have constant depth and
	\begin{equation}
		w=2^{O(D)}=2^{O(d)}=2^{O(2^n)}.
		\label{sm:eq:independent-width}
	\end{equation}

\end{document}